\documentclass[11pt]{article}

\usepackage[spanish,english]{babel} 
\usepackage{amssymb}
\usepackage{amsthm}

\usepackage{amsmath}

\newcommand{\eps}{{\varepsilon}}

\newcommand{\RR}{{\mathbf R}}

\newcommand{\NN}{{\mathbf N}}
\newcommand{\ZZ}{{\mathbf Z}}
\newcommand{\QQ}{{\mathbf Q}}
\newcommand{\CC}{{\mathbf C}}

\newcommand{\cO}{{\cal O}}

\newcommand{\cp}{{\cal P}}
\newcommand{\ca}{{\cal A}}

\newcommand{\ZXn}{{\ZZ[X_1,\ldots,X_n]}}

\newcommand{\algos}[6]{{
\begin{alg}\label{#2}\end{alg}
\noindent\textbf{Name of the algorithm:} #1\\
\textbf{Context:} #3\\
\textbf{Input:} #4\\
\textbf{Description of the algebraic computation tree:} #5\\
\textbf{Complexity analysis:} #6$\hfill \qed$}}

\def\squareforqed{\hbox{\rlap{$\sqcap$}$\sqcup$}}
\def\qed{\ifmmode\squareforqed\else{\unskip\nobreak\hfil
\penalty50\hskip1em\null\nobreak\hfil\squareforqed
\parfillskip=0pt\finalhyphendemerits=0\endgraf}\fi}

\newcommand{\AJ}{{\frak a}}

\renewcommand{\epsilon}{\varepsilon}

\newtheorem{thm}{Theorem}[section]
\newtheorem{cor}[thm]{Corollary}

\newtheorem{lem}[thm]{Lemma}

\newtheorem{prop}[thm]{Proposition}

\newtheorem{deft}[thm]{Definition}

\newtheorem{alg}[thm]{Algorithm}

\author{R.~Grimson\footnote{Departamento de Matem\'atica, Universidad de Buenos Aires and CONICET, Ciudad Universitaria, Pab. I, 1428 Buenos Aires, Argentina. email: rgrimson@dm.uba.ar}, 
J.~Heintz\footnote{Corresponding Author. Departamento de Computaci\'on, Universidad de Buenos Aires and CONICET,
Ciudad Universitaria, Pab. I, 1428 Buenos Aires, Argentina, and
Departamento de Matem\'aticas, Estad\'istica y
Computaci\'on, Facultad de Ciencias, Universidad de Cantabria, Avda. de
los Castros s/n, 39005 Santander, Spain. email: joos@dc.uba.ar $\&$
joos.heintz@unican.es}, B.~Kuijpers\footnote{Database and Theoretical Computer Science Research Group, Hasselt University, Agoralaan, Gebouw D, 3590 Diepenbeek, Belgium.
email: bart.kuijpers@uhasselt.be}}

\title{Evaluating geometric queries using few arithmetic operations}

\begin{document}
\maketitle

\begin{abstract}
Let $\cp:=(P_1,...,P_s)$ be a given family of $n$-variate polynomials with integer coefficients and suppose that the degrees and logarithmic heights of these polynomials are bounded by $d$ and $h$, respectively. Suppose furthermore that for each $1\leq i\leq s$ the polynomial $P_i$ can be evaluated using $L$ arithmetic operations (additions, subtractions, multiplications and the constants 0 and 1). Assume that the family $\cp$ is in a suitable sense \emph{generic}. We construct a database $\cal D$, supported by an algebraic computation tree, such that for each $x\in [0,1]^n$ the query for the signs of $P_1(x),...,P_s(x)$ can be answered using $h d^{\cO(n^2)}$ comparisons and $nL$ arithmetic operations between real numbers. The arithmetic-geometric tools developed for the construction of $\cal D$ are then employed to exhibit example classes of systems of $n$ polynomial equations in $n$ unknowns whose consistency may be checked using only few arithmetic operations, admitting however an exponential number of comparisons.
\end{abstract}

\paragraph{Keywords:} Constraint database, query evaluation, computational complexity, consistency of polynomial equation systems.

\paragraph{MSC:} 68P15, 68Q25, 14P99, 14Q99.

\section{Introduction}\label{SecArithIntro}

By $\NN, \ZZ, \QQ$ and $\RR$ we denote the natural, the integer, the rational and the real numbers. Let $X_1,...,X_n$ be indeterminates over $\RR$ and let $X:=(X_1,...,X_n)$. Furthermore let $\cp:=(P_1,...,P_s)$ be a family of polynomials of $\ZZ[X]$. A sign condition $\sigma\in\{-1,0,1\}^s$, with $\sigma=(\sigma_1,...,\sigma_s)$, determines a polynomial
inequality system of the form
\begin{equation}\label{Eq1} \bigwedge_{1\leq i\leq s}
sign(P_i)=\sigma_i.\end{equation} We call $\sigma$ \emph{consistent}
if there exists a point $x\in\RR^n$ satisfying Condition~\ref{Eq1}.

The consistent sign conditions on the family $\cp=(P_1,\ldots,P_s)$
define a semi-algebraic partition $\ca:=\ca(\cp)$ of $\RR^n$, called an \emph{arrangement}. This arrangement induces a query $Q_{\ca}$ which determines, for each $x\in\RR^n$ the
(unique) element $A\in\ca$ with $x\in A$. This query is called the \emph{sign
condition query} for $P_1,\ldots,P_s$.

In 1984, Meyer auf der Heide~\cite{Mayer84Knapsack} presented a fast solution 
to the sign condition query for a family of linear polynomials with integer
coefficients (see also Theorem 3.27 in \cite{BurCS}). His strategy was to find, given a family of affine
linear polynomials in $\ZXn$ of logarithmic height bounded by
$h\in\NN$, a real number $\delta>0$ depending only on $n$ and $h$, called
the \emph{coarseness} of the given family, such that any hypercube
in $\RR^n$ of side-length $\delta$ has the following property: all the
affine hyperplanes defined by the given family that cut the
hypercube have a common intersection point. This fact allowed him to
design a non-uniform polynomial-time algorithm which solves the
Knapsack Problem. 

If the given family of linear polynomials is \emph{generic} in the sense defined below, then no $n+1$ polynomials have a common zero. Thus, any hypercube in $\RR^n$ of side-length $\delta$ satisfies a stronger condition: it is cut at most $n$ hyperplanes defined by the given family. In view of this observation we shall extend Meyer auf der Heide's argumentation to semi-algebraic subsets of $\RR^n$ defined by certain generic families of polynomials of arbitrary degree.

\paragraph{Notations and assumptions.}
Let $X:=(X_1,...,X_n)$ and let $P\in\ZZ[X]$ be a
polynomial; we denote by $\deg(P)$ its total degree and by $h(P)$ its logarithmic height defined as the maximum of the logarithmic heights of its coefficients. Given a semi-algebraic set $S$ in $\RR^n$, if the sign of $P$ is not constant on $S$, we say that $P$ \emph{cuts} $S$.

Let $\cp:=(P_1,...,P_s)$ be a family of distinct polynomials in
$\ZZ[X]$ of total degree and logarithmic height
bounded by $d$ and $h$ respectively. We remark that this
implies a bound on the cardinality $s$ of $\cp$: since each
polynomial is determined by its $\binom{n+d}{d}$ coefficients and
each coefficient by its $h$ bits, we obtain $s\leq
2^{h\binom{n+d}{d}}\leq 2^{h(d+1)^n}$.

We further remark that this bound on the cardinality of $\cp$ makes
the parameter $s$ disappear from the complexity bounds we are going to exhibit in this paper. Since $s$ is bounded by $2^{h(d+1)^n}$, a bound of the kind $(hd^n\log(s))^{\cO(1)}$ equals $(hd^n)^{\cO(1)}$, where $s$ does not intervene.

Given a first-order formula in the language of real closed fields $\varphi(X_1,...,\allowbreak X_n)$ with $n$ free-variables, we denote by ${\cal R}(\varphi)$ the \emph{realization} of $\varphi$ in $\RR^n$, namely ${\cal R}(\varphi):=\{x\in\RR^n\;|\;\RR\models\varphi(x)\}$. Analogously, for a given polynomial $P\in\ZZ[X]$, we denote by  ${\cal R}(P)$ the set 
${\cal R}(P):=\{x\in\RR^n\;|\;P(x)=0\}$.


We shall say that the family $\cp=(P_1,...,P_s)$ is
\emph{generic} if for any $1\leq r\leq n$ and any
$1\leq i_1<...<i_r\leq s$, the polynomials $P_{i_1},...,P_{i_r}$
form a regular sequence in $\QQ[X]$ or generate the
trivial ideal.
Since this notion of genericity is too restrictive for the results of the second part of this paper, we define a weaker notion. 

Let $\Delta\in\ZZ[X]$ be a polynomial of total degree and logarithmic height
bounded by $d$ and $h$ respectively. Given any semi-algebraic set $S\subset \RR^n$, we denote by $S_{\Delta}$ the set $$S_{\Delta}:=\{x\in S\;|\;\Delta(x)\geq 1\}.$$
We remark that for any $\Delta$ and any closed $S$ the set $S_{\Delta}$ is closed. 

\begin{deft}\label{DefGenAss} We say that the family $(P_1,...,P_s)$ is
\emph{generic outside ${\cal R}(\Delta)$} if for any $1\leq r\leq n$ and any
$1\leq i_1<...<i_r\leq s$, the polynomials $P_{i_1},...,P_{i_r}$
form a regular sequence in $\QQ[X]_{\Delta}$ or generate the
trivial ideal.
\end{deft}

We observe that both notions of genericity agree for $\Delta=1$.

\paragraph{Model of Computation.} 
In this paper we analyse the complexity of the sign condition problem from a constraint database point of view. We use algebraic computation trees to represent
constraint databases (see \cite{CDBs,Strassen90,BurCS} for the
background on constraint databases and algebraic complexity theory). We measure the number of arithmetic operations and comparisons that our algorithms perform on a given input. While divisions are not allowed in our model, the arithmetic operations of addition, subtraction and multiplication can be performed among any two real numbers. In particular all straight-line programs we are going to consider in this paper are supposed to be division-free. We remark that if we restrict our attention to rational inputs, the bit complexity of any algorithm in this paper can be bounded by a polynomial depending on its algebraic complexity and the logarithmic height of the input, \textit{i.e.}, it remains controlled.

\paragraph{Outline of the article.}
Let $\Delta,P_1,...,P_s$ polynomials in
$\ZZ[X]$ of total degree and logarithmic height bounded by
$d>1$ and $h$ respectively and assume that the family $\cp:=(P_1,...,P_s)$ is generic outside ${\cal R}(\Delta)$. We decompose the hypercube $[0,1]^n\subset\RR^n$ in small closed hypercubes with mutually disjoint interiors and side length of order $2^{-h d^{\cO(n^2)}}$.  We shall prove that if $R$ is one of these  hypercubes, the set $R_{\Delta}$  is cut by at most $n$ polynomials in $\cp$; the remaining polynomials in $\cp$ have constant, namely positive or negative, sign on it. The indices of the polynomials of $\cp$ that cut $R_{\Delta}$ for each  hypercube $R$ and the signs of the remaining polynomials on it are stored, at a preprocessing stage, in a constraint database.

In Section \ref{SecBit}, we present an algorithm that solves the sign condition
problem for this family in the set $[0,1]^n_{\Delta}=\{x\in [0,1]^n\;|\;\Delta(x)\geq 1\}$. It works as follows. Given a point $x\in [0,1]^n_{\Delta}$, the
pre-constructed database allows us to determine a small hypercube $R$ containing the point $x$, the polynomials in $\cp$ that cut the set $R_{\Delta}$ and the (constant) signs of the remaining polynomials in $\cp$ performing $h d^{\cO(n^2)}$ comparisons. Evaluating at $x$ the (at most $n$) polynomials of $\cp$ that cut the set $R_{\Delta}$, we solve the sign condition problem for the family $\cp$ performing only  $d^{\cO(n)}$ arithmetic operations. 

Finally, in Section \ref{SecApp}, we discuss applications of our method. Our results show that, admitting branchings in our complexity model, the number of algebraic operations necessary to answer the query for the consistency of two classes of simply structured equation systems of $n$ polynomials in $n$ unknowns may drop dramatically with respect to traditional methods based on the evaluation of elimination polynomials or their coefficient representation.

In \cite{GHK11} we have shown, without any genericity assumption, that the sign
condition query for $P_1,...,P_s$ can be evaluated using only
$\cO((L+n)^5\log(s))$ arithmetic operations and comparisons, where $L$ is the
number of essential multiplications (additions and scalar operations are
free) necessary to evaluate {\em all} polynomials $P_1,...,P_s$. However
this result cannot be applied in Section~\ref{SecApp}.



\section{The sign condition algorithm for polynomials with integer coefficients}\label{SecBit}

\subsection{Mathematical tools}
\paragraph{The Arithmetic Relation between Distance and Height.}
We extend the notion of height to rational numbers. If
$r=\frac{p}{q}\neq 0$ is a rational number with $p,q\in\ZZ$ coprime,
we define its \emph{height} as $H(r):=\max\{H(p),H(q)\}$ and its
\emph{logarithmic height} as $h(r):=\max\{h(p),h(q)\}$.

From Theorem 1.3.1 in \cite{BaPoRo94Paper} (see alternatively Theorem
14.21 in \cite{BaPoRo06}) we immediately obtain the following
result.

\begin{prop}\label{CorBPRJACM95}
Let $\varphi(Y_1,...,Y_l)=\exists X_1...\exists X_k
\psi(X_1,...,X_k,Y_1,...,Y_l)$ be an existential first-order formula with $l$ free variables, where $\psi$ is a boolean combination of polynomial equalities and inequalities involving polynomials with integer coefficients of degree bounded
by $d$ and logarithmic height bounded by $h$. Then, there exists
an equivalent quantifier-free formula involving polynomials with integer coefficients of
degree bounded by $d^{\cO(k)}$ and logarithmic height bounded by
${h d^{\cO(kl)}}$.\qed
\end{prop}

The following lemma is a form of Cauchy's bound on the roots of a univariate polynomial.

\begin{lem}\label{ALem1}
Let $\psi$ be a quantifier-free first-order formula over the reals
with only one free variable and involving polynomials with integer
coefficients of height bounded by $H$. Let $\mu\neq0$ belong to the border of the
realization of $\psi$, ${\cal{R}}(\psi)\subset\RR$. Then,
$|\mu|>\frac{1}{H+1}$.
\end{lem}

\begin{proof}
Since the formula $\psi$ is quantifier free, $\mu$ must be a root of
some of the polynomials involved in $\psi$, say
$P=\sum_{i=0}^da_iT^i\in\ZZ[T]$. The height of $P$ is bounded by
$H$. We assume, without loss of generality, $a_0\neq 0$.

Thus, by Cauchy's bound for the zeroes of a polynomial \cite{mignotte92}, the absolute
value of any root of $P$ is at least $\frac{1}{H+1}$. In particular,
$|\mu|>\frac{1}{H+1}$.
\end{proof}

\begin{deft}
For $\alpha_1,\alpha_2,...,\alpha_n,\delta\in\QQ$, $\delta>0$,
the set
$$[\alpha_1,\alpha_1+\delta]\times...\times[\alpha_n,\alpha_n+\delta]\subset
\RR^n$$ is called a \emph{rational hypercube} with \emph{coordinates}
$\alpha_1,...,\alpha_n,\alpha_1+\delta,...,\alpha_n+\delta$. The \emph{logarithmic
height} of the hypercube is defined as the maximal logarithmic height of its
coordinates. Its \emph{side length} is $\delta$.
\end{deft}

\begin{prop}\label{ArithProp}
There exists a universal real constant $c>0$ that satisfies, for
any three positive integers $h, n$ and $d>1$ the following
condition.

Let $R\subset\RR^n$ be a rational hypercube of logarithmic
height $h$ and let $P_1,...,P_r,Q,\Delta\in \ZZ[X]$ be
polynomials of degree at most $d$ and of logarithmic height bounded
by $h$. Consider the subsets of $\RR^n$
$$S:=\{P_1=0,...,P_r=0\}\cap R_{\Delta}\;\;\text{and}\;\;
T:=\{Q=0\}\cap R_{\Delta}.$$ If $S$ and $T$ are
disjoint then their Euclidean distance is strictly greater than
$2^{-h d^{cn}}$.
\end{prop}

\begin{proof}
Let us assume that $S$ and $T$ are disjoint and
non-empty. Since they are closed and bounded, the distance between
them is positive, say $\mu\in\RR_{>0}$. Denoting by $\chi_{R_{\Delta}}$ the canonical
first-order formula that expresses the set $R_{\Delta}$, we consider the
following first-order formula with $\eps$ as its only free-variable:
$$\displaylines{\quad \varphi(\varepsilon):=\exists
X_1\;...\;\exists X_n \exists Y_1\;...\; \exists
Y_n\;\;\;\chi_{R_{\Delta}}(X)\;\land\;\chi_{R_{\Delta}}(Y)\;\land
\hfill{}\cr\hfill{}P_1(X)=0\;\land\;\;...\;\;\land\; P_r(X)=0\;\land
\hfill{}\cr\hfill{} Q(Y)=0\;\land\;|X-Y|^2\leq
\varepsilon^2.\quad\quad\quad\quad\quad\quad\quad\;\;\;}$$

Clearly, a positive real number $\varepsilon$ satisfies this formula
if and only if $\varepsilon\geq \mu$. The formula $\varphi$ contains
$\eps$ as the only free variable, $X_1,...,X_n,Y_1,...,Y_n$ as
bounded variables and involves polynomials of degree at most $d$ and
logarithmic height at most $h$. Hence, by Proposition
\ref{CorBPRJACM95}, there exists an equivalent quantifier-free
first-order formula $\psi$ involving polynomials of degree bounded
by $d^{\cO(n)}$ and height bounded by $2^{h d^{\cO(n)}}$. Thus,
by Lemma \ref{ALem1}, $\mu$ is at least $\frac{1}{1+2^{h
d^{\cO(n)}}}=2^{-h d^{\cO(n)}}$.
\end{proof}

\paragraph{Semi-algebraic $k$-determined sets.}
Given a semi-algebraic set $S\subset
\RR^n$ and a natural number $1\leq k\leq n$, we say that the family $\cp$ is \emph{$k$-determined over $S$} if at most $k$ polynomials of the family $\cp$ cut $S$ (in other words, the signs of all the polynomials of $\cp$ are determined over $S$, except for at most $k$ of them).

The next statement shows that under our genericity assumption (see Definition 
\ref{DefGenAss}) the family $\cp$ is $n$-determined over $R_{\Delta}$ for any hypercube $R\subset[0,1]^n$ of side length $2^{-h d^{\cO(n^2)}}$.

\begin{prop}\label{PropDelta}
There exists a universal real constant $c'>0$ that satisfies, for
any three positive integers $h, n$ and $d>1$ the following
condition.

Let $\Delta,P_1,...,P_s$ be polynomials in
$\ZZ[X]$ of total degree and logarithmic height bounded by
$d>1$ and $h$ respectively. Assume that the family $\cp:=(P_1,...,P_s)$ is generic outside ${\cal R}(\Delta)$. Then, given any hypercube $R\subset [0,1]^n$ of side length $\delta:=2^{-h d^{c' n^2}+1}$, the family $\cp$ is $n$-determined over the set $R_{\Delta}$.
\end{prop}

In the spirit of~\cite{Mayer84Knapsack}, we call the rational number $\delta$ a \emph{coarseness} for the family $\cp$.

\begin{proof}
Let $c$ be the universal constant of Proposition \ref{ArithProp} and let $c':=c+1$. Remark that the following condition is satisfied
\begin{equation}\label{eqeq}h\frac{d^{n^2(c'-c)}}{\ulcorner\log(n)\urcorner^n}\geq 1.\end{equation}

Let $R^0:=R\subset [0,1]^n$ be an hypercube of side length $\delta:=2^{-h d^{c' n^2}+1}$. We will show that the family $\cp$ is $n$-determined over the set $R_{\Delta}$. Assume to the contrary and with out loss of generality that $P_1,...,P_{n+1}$ are different polynomials of $\cp$ that cut the set $R_{\Delta}$.

For $1\leq i\leq n+1$, we define inductively a hypercube $R^i$ with the following properties:
\begin{itemize}
\item its logarithmic height is $h_i:=h\frac{d^{c'n^2-cn (i-1)}}{\ulcorner \log(n)\urcorner^{i-1}}$,
\item its side length is $2^{-h_i}<1$,
\item $R^{i-1}\subset R^{i}\subset [0,1]^n$,
\item the polynomials $P_{1},...,P_{n+1}$ cut the set $R^i_{\Delta}$ and
\item the polynomials $P_1,...,P_i$ have a common intersection point in $R^i_{\Delta}$.
\end{itemize}

Let $h_1:={h d^{c' n^2}}$ and let $R^1\subset [0,1]^n$ be an hypercube of side length $\delta_1:=2^{-h_1}=2\delta$ and logarithmic height $h_1$ that contains $R^0$. Since  $P_1,...,P_{n+1}$ cut $R^0_{\Delta}$, the polynomial $P_1$ cuts also $R^1_{\Delta}$. 

For $1\leq i\leq n$, assume that a hypercube $R^i\supset R^{i-1}$ of side length $2^{-h_i}$ satisfying the conditions stated above has been defined. In particular, the polynomials $P_1,...,P_i$ have a common intersection point in $R^i_{\Delta}$ and $P_{i+1}$ cuts $R^i_{\Delta}$.

Let $R^{i+1}$ be a hypercube contained in $[0,1]^n$ of logarithmic height $h_{i+1}$ and side length $2^{-h_{i+1}}$ that contains the hypercube $R^i$. By construction, it satisfies the first four conditions stated above. We claim that the polynomials $P_1,...,P_{i+1}$ have a common intersection point in $R^{i+1}_{\Delta}$.

Since $R^{i+1}\supset R^{i}$, the induction hypothesis implies that $S:=\{P_1=0,...,P_i=0\}\cap R^{i+1}_{\Delta}$ and $T:=\{P_{i+1}=0\}\cap R^{i+1}_{\Delta}$ are not empty. 

By Proposition \ref{ArithProp}, if $S$ and $T$ are disjoint then their Euclidean distance is strictly greater than $\rho_i:=2^{-h_{i+1} d^{c\cdot n}}$. Since both sets have points inside $R^{i}_{\Delta}$ and the diagonal of the hypercube $R^{i}$ measures $\sqrt n 2^{-h_i}\leq 2^{-h_i+\log (n)} \leq 2^{-h_i/ \log(n)} \leq \rho_i$ we conclude that $S$ and $T$ are not disjoint. Thus the polynomials $P_1,...,P_{i+1}$ have a common intersection point in $R^{i+1}_{\Delta}$ what concludes the induction step.

After $n$ steps, we obtain a hypercube $R^{n+1}\subset[0,1]^n$ such that $P_1,...,P_{n+1}$ have a common intersection point in $R^{n+1}_{\Delta}$. This contradicts the genericity of the family $\cp$ outside ${\cal R}(\Delta)$. We conclude that  the family $\cp$ is $n$-determined over the set $R_{\Delta}$.
\end{proof}

We remark that in the case of affine hyperplanes ($d=1$, the affine
linear case) Proposition \ref{CorBPRJACM95}
and hence Proposition \ref{ArithProp} can be improved to obtain
logarithmic heights and distances which are simply exponential in the dimension $n$. However, the logarithmic height resulting from the inductive proof of Proposition \ref{PropDelta} is of order $2^{h
\cO(n)^n}$ and therefore doubly exponential in the dimension $n$. Thus, our method does not represent a proper generalization of Meyer auf der Heide's method for the linear case.

\subsection{The Sign Condition Algorithm} 
In this section we assume given $\Delta,P_1,...,P_s$ polynomials in
$\ZZ[X]$ of total degree and logarithmic height bounded by
$d>1$ and $h$ respectively. We further assume that the family $\cp:=(P_1,...,P_s)$ is generic outside ${\cal R}(\Delta)$.

Based on the results of the last
paragraph, we construct a big database $\cal D$ of size $2^{h
d^{\cO(n^2)}}$ that allows, for any $x\in[0,1]^n_{\Delta}$, a fast
determination of the sign condition satisfied by the polynomials of
$\cp$ at $x$. 


First we describe the following algorithm which is a form of
multidimensional binary search, with $m$ as a parameter.

\algos{${\cal D}_0^m$}{Query-Partition}{There are fixed positive integers $m$ and $n$.}{A point
$x=(x_1,...,x_n)\in[0,1]^n$.}{\\The algebraic computation tree performs, for
$j=1,...,n$, a unidimensional binary search to determine $i_j$, defined as the first $i$ such that
$x_j\in [\frac{i}{m}, \frac{i+1}{m}]$. \\ The computation finishes in a leaf labeled $(i_1,...,i_n)\in\NN^n$ such that
$x$ belongs to $\Pi_{j=1}^n[\frac{i_j}{m}, \frac{i_j+1}{m}]$.
} {Let $\mu=h(m)$ be the logarithmic height of $m$. Each unidimensional binary search requires $\cO(\mu)$ comparisons. 
Hence, the algorithm has algebraic complexity $\cO(\mu n)$.} 

\medskip

See for instance~\cite{Knuth3} for the definition and complexity analysis of the unidimensional binary search.
We remark that the algebraic computation tree  ${\cal D}_0^m$ uses only comparisons and no arithmetic operations. 

%


\medskip

A direct instantiation of the decision procedure for the existential theory of the reals given in \cite{BaPoRo06}, leads to the following result, which is useful for the cost analysis of the construction of the database $\cal D$ below.

\begin{prop}\label{PropCuts}
Let $R\subset\RR^n$ be a hypercube, let $P,\Delta\in\ZZ[X]$ be 
polynomials of degree at most $d>1$ and suppose that $h$ bounds the logarithmic heights of
$R$, $P$ and $\Delta$. Then, it is possible to decide whether $P$ cuts $R_{\Delta}$ using
$d^{\cO(n)}$ arithmetic operations between rational numbers of logarithmic
height bounded by $h{d^{\cO(n)}}$.
\end{prop}

\begin{proof}\label{PropCuts?}
Let $\chi_{R_{\Delta}}$ the canonical first-order formula that describes the set ${R_{\Delta}}$. Clearly, the formula $$\varphi:=\exists X\; (P(X)=0\land
\chi_{R_{\Delta}}(X))$$ is true if and only if $P$ cuts $R_{\Delta}$. By Theorems 13.11 and 13.13 
in \cite{BaPoRo06}, this sentence can be decided within the stated bounds.
\end{proof}


\begin{thm}\label{ThmSCArith}
Let $\Delta,P_1,...,P_s$ be polynomials in
$\ZZ[X]$ of total degree and logarithmic height bounded by
$d>1$ and $h$ respectively. Assume that the family $\cp:=(P_1,...,P_s)$ is generic outside ${\cal R}(\Delta)$ and that for any $1\leq k\leq s$ there is given a  straight-line program $\beta_k$ in $\ZZ[X]$ of length $L\leq d^{\cO(n)}$ which evaluates the polynomial $P_k\in\cp$.

Then, there exists a database $\cal D$, represented by an algebraic computation tree of size $2^{h d^{\cO(n^2)}}$,
that can be constructed in time $2^{h d^{\cO(n^2)}}$ and that
allows to determine, for any point $x\in[0,1]^n_{\Delta}$, 
the sign conditions satisfied by the polynomials of $\cp$ at $x$, performing 
$h d^{\cO(n^2)}$ comparisons and $nL=d^{\cO(n)}$ arithmetic operations in $\RR$.
\end{thm}

\begin{proof}
By Proposition \ref{PropDelta}, there exists a universal constant $c\in\NN$
such that for $m:=2^{h d^{cn^2}}$ and
$\delta:=\frac{1}{m}$ the family of polynomials $\cp$ is $n$-determined over $R_{\Delta}$ for any hypercube
$R\subset [0,1]^n$ of side length $\delta$. With other words, $\delta$
is a coarseness for the family $\cp$. We assume that the constant
$c$ is known.

Based on the Algorithm \ref{Query-Partition} the database $\cal D$ becomes constructed as an algebraic computation tree obtained from ${\cal D}_0^m$ in the following way: to a leaf $l$ of ${\cal D}_0^{m}$ as above with label $(i_1,...,i_n)$, a straight-line program is added. The straight-line program evaluates at the input of $\cal D$ the (at most $n$) polynomials of $\cp$ that cut the set $(\Pi_{j=1}^n[\frac{i_j}{m}, \frac{i_j+1}{m}])_{\Delta}$ and outputs $(i_1,...,i_n)$ and the sign condition $\sigma$ resulting from these evaluations. The cost of the construction of ${\cal D}_0^{m}$ is proportional to its size: $2^{h d^{\cO(n^2)}}$. Then we determine, for each $1\leq k\leq s$, if the polynomial $P_k$
cuts the hypercube $\Pi_{j=1}^n[\frac{i_j}{m}, \frac{i_j+1}{m}]$ following Proposition \ref{PropCuts}. In this way, the complete construction of the
algebraic computation tree $\cal D$ requires $sm^n d^{\cO(n)}=2^{h d^{\cO(n^2)}}$
arithmetic operation between rational numbers of logarithmic height
bounded by $h{d^{\cO(n^2)}}$. 

We claim that $\cal D$ computes a partition of $[0,1]^n_{\Delta}$ which is finer than the partition induced by the sign conditions realized by the family $\cp$. To prove this
claim let $x,x'\in [0,1]^n_{\Delta}$ be such that $\cal D$ outputs the same
result $(i_1,...,i_n), \sigma$ for both $x$ and $x'$. Hence, $x$ and
$x'$ belong to the same small hypercube
$R=\Pi_{j=1}^n[\frac{i_j}{m}, \frac{i_j+1}{m}]$ and satisfy the same
sign condition $\sigma$ for the polynomials that cut $R_{\Delta}$. Since the
remaining polynomials in $\cp$ do not cut $R_{\Delta}$, their signs remain
invariant on it. In particular, $\cp$ realizes the same sign
condition at $x$ and at $x'$.

Finally, let us analyze the complexity of querying the database $\cal D$. Following the complexity analysis of Algorithm \ref{Query-Partition} the binary search requires $ h d^{\cO(n^2)}$ comparisons. The evaluation of a polynomial $P_k\in \cp$, $1\leq k\leq s$ using the straight-line program $\beta_k$ costs $L$ arithmetic operations. Since at most $n$ polynomials of $\cp$ cut the set $R_{\Delta}$, the complete evaluation phase costs at most $nL\leq n d^{\cO(n)}=d^{\cO(n)}$ arithmetic operations. Thus,  the algebraic computation tree $\cal D$ determines, for any input point $x\in[0,1]^n_{\Delta}$, 
the sign conditions satisfied by the polynomials of $\cp$ at $x$ performing 
$h d^{\cO(n^2)}$ comparisons and $nL=d^{\cO(n)}$ arithmetic operations.
\end{proof}


\section{Applications}\label{SecApp}

Let $X_1,...,X_n$ and $U_1,...,U_m$ be indeterminates over $\QQ$ and let $X:=(X_1,...,\allowbreak X_n)$ and $U:=(U_1,...,\allowbreak U_m)$.

Let $d,h\in \NN$ with $d>1$  and suppose that there are given polynomials $G_1,...,G_n\in\ZZ[X]$ of degree and logarithmic height at most $d$ and $h$, which form a regular sequence in $\QQ[X]$. Let $\xi_1,...,\xi_s\in\ZZ^n$ be the integer zeros of the polynomial equation system $G_1=0,...,G_n=0$.

From the geometric and the arithmetic B\'ezout Inequalities \cite{Heintz83}, \cite{Fulton84}, \cite{Vogel84} and \cite{BoGiSo94}, \cite{Philippon91,Philippon94,Philippon95} we deduce $s\leq d^n$ and that the sum of the logarithmic heights of $\xi_1,...,\xi_s$ is bounded by $(h+2n^2)d^n$ (see also [KriPaSo01], Section 1.2).

Finally let $H\in\ZZ[U,X]$ and $\Delta\in\ZZ[U]$ be polynomials of degree and logarithmic height at most $d$ and $h$ respectively.

We assume that the family of polynomials ${\cal H}:=\{H(U,\xi)\;|\;\xi\in\ZZ^n,\allowbreak  G_1(\xi)=0,...,G_n(\xi)=0\}$ is generic in the following sense: for any $1\leq r\leq n$ and any $1\leq i_1<...<i_r\leq s$ the polynomials $H(U,\xi_{i_1}),...,H(U,\xi_{i_r})$ form a regular sequence in $\QQ[U]_{\Delta}$ or generate the trivial ideal.

Observe that for each $1\leq i\leq s$ the degree and logarithmic height of $H(U,\xi_i)$ is bounded by $d$ and $(hm)d^{\cO(n)}$, respectively.

Therefore, we may, for $\delta:=2^{-hd^{\cO(m^2+n)}}$, apply the Proposition \ref{PropDelta} to the family $\cal H$. Thus any hypercube $R\subset[0,1]^m$ of side length $\delta$ has the following property:
there exist $0\leq r\leq m$ indices $1\leq i_1<...<i_r\leq s$ such that $H(U,\xi_{i_1}),...,\allowbreak H(U,\xi_{i_r})$ are the only polynomials of $\cal H$ which  change their signs in $R_{\Delta}$. According to Theorem \ref{ThmSCArith} this implies the following complexity result.

\begin{thm}\label{Thm21}
Let notations and assumptions be as before and suppose that the polynomial $H$ can be evaluated by a straight-line program $\beta$ of size $L$ in $\ZZ[U,X]$. Then there exists a database $\cal D$ represented by an algebraic computation tree of size $(mL)2^{hd^{\cO(m^2+n)}}$ such that for any point $u\in [0,1]^m_{\Delta}$ the query whether the polynomial equation system $G_1(X)=0,...,G_n(X)=0,H(u,X)=0$ has an integer solution can be evaluated in the database $\cal D$ performing $hd^{\cO(m^2+n)}$ comparisons and $mL$ arithmetic operations in $\RR$.
\end{thm}

\begin{proof}
By Algorithm \ref{Query-Partition} there exists a decision tree ${\cal D}_0$ (see Algorithm \ref{Query-Partition}) of size $2^{hd^{\cO(m^2+n)}}$ which for any $u\in[0,1]^m$ determines a hypercube $R\subset [0,1]^m$  of side-length $\delta$ to which $u$ belongs, using for this task $hd^{\cO(m^2+n)}$ comparisons in $\RR$.  Each leaf of ${\cal D}_0$ corresponds to such an hypercube $R$ and determines therefore an integer $1\leq r\leq m$ and a set of indices $1\leq i_1<...<i_r\leq s$ as before. The algebraic computation tree ${\cal D}$ is obtained from ${\cal D}_0$ in the following way: to a leaf $l$ of ${\cal D}_0$ as above, a straight-line program is added. This program consists of $r$ sequential copies of the labeled directed acyclic graph of $\beta$  with $X$ instantiated to $\xi_{i_1},...,\xi_{i_r}$. The output nodes of these copies of $\beta$ are followed by decision nodes connected in a way such that the algebraic computation tree accepts an input $u\in[0,1]^m_{\Delta}$ which belongs to the hypercube $R$ above if there exists an index $1\leq k\leq r$ with $H(u,\xi_{i_k})=0$. This is exactly the case when the polynomial equation system $G_1(X)=0,...,G_n(X)=0,H(u,X)=0$ has an integer solution.
\end{proof}
\section{Examples}

Let $X:=(X_1,...,\allowbreak X_n)$, $U:=(U_1,...,\allowbreak U_m)$ and $d,h\in\NN$ with $d>1$ be as in Section \ref{SecApp}.

\subsection{Example 1}
Suppose $d\geq m$ and let $H\in \ZZ[U,X]$ be a polynomial of degree and logarithmic height at most $d$ and $h$. Let $\deg_X(H) \geq 1$ and let $\Delta$ be a non-zero coefficient of $H$ belonging to a monomial in $X_1,...,X_n$ of degree not less than one. Then the logarithmic height and the degree of $\Delta$ are also bounded by $d$ and $h$.

Suppose the $H$ can be evaluated by a straight-line program $\beta$ of size $L$ in $\ZZ[U,X]$.

For $1\leq i\leq d$ let $Y^{(i)}_1,...,Y^{(i)}_n$ be new indeterminates over $\QQ$ and let $Y_i:=(Y^{(i)}_1,...,Y^{(i)}_n)$. There exist polynomials $\Gamma_1,...,\Gamma_n\in\ZZ[Y_1,...,Y_d,X]$ of degree and logarithmic height at most $(nd)^{\cO(n)}$, such that the polynomial equation system $\Gamma_1=0,...,\Gamma_n=0$ defines over the algebraic closure of $\QQ(Y_1,...,Y_d)$ the points $Y_1,...,Y_d$.

This can be seen as follows. The product $\AJ$ of the ideals of $\QQ(Y_1,...,Y_d)$ generated by the polynomials $X_1-Y_1^{(i)},...,X_n-Y_n^{(i)}, 1\leq i\leq d$, has only the zeroes $Y_1,...,Y_d$ in the algebraic closure of $\QQ(Y_1,...,Y_d)$. Therefore, the ideal $\AJ$ is zero-dimensional and generated by polynomials of the form $\prod_{1\leq i\leq d} X_{k_i}-Y_{k_i}^{(i)}$, $1\leq k_1,...,k_d\leq n$ which have all degree $d$. Applying to these polynomials a suitable effective version of the Shape Lemma (see e.g. \cite{KrPa96}, Theorem 2.2 and its proof), one obtains the polynomials $\Gamma_1,...,\Gamma_n$ as generators of the radical ideal of $\AJ$.

\begin{lem}\label{lemEx1}
Let $1\leq k\leq m$. Then there exists a non-empty Zariski open set ${\cal O}_k$ of $\CC^{k\times n}$ such that for any $y\in{\cal O}_k$ with $y:=(y_1,...,y_k)$ and $y_1,...,y_k\in\CC^n$ the locally closed subvariety of $\CC^m$ defined by the conditions $H(U,y_1)=0,...,H(U,y_k)=0$ and $\Delta(U)\neq 0$ is empty or equidimensional of dimension $m-k$.
\end{lem}

\begin{proof}
In $\CC^{m}\times \CC^{k\times n}$ we consider the locally closed subvariety $W$ defined by the conditions

$$H(U,Y_1)=0,...,H(U,Y_k)=0\text{  and  }\Delta(U)\neq 0.$$ 

Let $\pi_1:\CC^{m}\times \CC^{k\times n}\rightarrow \CC^{m}$ and $\pi_2:\CC^{m}\times \CC^{k\times n}\rightarrow \CC^{k\times n}$ be the canonical projections and let $C$ be an irreducible component of $\overline W$. Then $\overline{\pi_1(C)}$ is an irreducible closed subvariety of $\CC^m$ and hence of dimension at most $m$. Moreover $\Delta$ does not vanish identically on $\overline{\pi_1(C)}$. Therefore we may choose generically a point $u\in\pi_1(C)$ with $\Delta(u)\neq 0$. Observe that $\pi_1^{-1}(u)\cap C$ is a closed subvariety of $C$ of dimension at most $k(n-1)=nk-k$. The Theorem on the Dimension of the Fibers \cite[Theorem 7]{Shafarevich} implies now $\dim (C)\leq m-k+nk$. Suppose first that $\overline{\pi_2(C)}$ is strictly contained in $\CC^{k\times n}$. Then ${\cal O}_{k,C}:=\CC^{k\times n}- \overline{\pi_2(C)}$ is a non-empty Zariski open subset of $\CC^{k\times n}$ and for $y\in {\cal O}_{k,C}$ we have $\pi_2^{-1}(y)\cap C = \emptyset$. 

Now suppose $\overline{\pi_2(C)}=\CC^{k\times n}$.
By the Theorem on the Dimension of the Fibers there exists a non-empty Zariski open subset 
${\cal O}_{k,C}\subset\CC^{k\times n}$ such that for any $y\in{\cal O}_{k,C}$ the variety $ \pi_2^{-1}(y)\cap C$ is empty or equidimensional of dimension $\dim (C)-kn\leq m-k$.
Let $${\cal O}_{k}=\bigcap_{\begin{array}{c}C \text{ irreducible} \\\text{component of }\overline{W}\end{array}}{\cal O}_{k,C}.$$ Then ${\cal O}_{k}$ is a non-empty algebraic subset of $\CC^{k\times n}$ and for any $y\in {\cal O}_{k}$ the irreducible components of $\pi_2^{-1}(y)\cap \overline{W}$ are of dimension at most $n-k$.

Let $y\in {\cal O}_k$ with $y=(y_1,...,,y_k)$ and $y_1,...,y_k\in\CC^n$ and let $D$ be an irreducible component of $\pi_2^{-1}(y)\cap W$. Observe that $D$ is an irreducible component of a locally closed subvariety of $\CC^m\times\CC^{k\times n}$ which is definable by $k(n+1)$ equations and the open condition $\Delta(U)\neq 0$. 

Therefore, $D$ has dimension at least $m-k$. From our previous argumentation we conclude now that the dimension of $D$ is \emph{exactly} $m-k$. In particular $\pi_2^{-1}(y)\cap W$ is empty or equidimensional of dimension $m-k$. Since $\pi_2^{-1}(y)\cap W$ is isomorphic to the locally closed subvariety of $\CC^m$ defined by the conditions $H(U,y_1)=0,...,H(U,y_k)=0$ and $\Delta(U)\neq 0$, the lemma follows.
\end{proof}

Lemma \ref{lemEx1} implies the following result.

\begin{cor}\label{corEx1}
There exists a non-empty Zariski open subset $\cal O$ of $\CC^{d\times n}$ such that any point $y\in \cal O$ with $y=(y_1,...,y_d)$ and $y_1,...,y_d\in \CC^n$ satisfies the following condition:

for any $1\leq r \leq m$ and any $1\leq i_1<...<i_r\leq d$ the polynomials $H(U,y_{i_1}),...,H(U,y_{i_r})$ form a regular sequence or generate the trivial ideal in $\QQ[U]_{\Delta}$.
\end{cor}

From \cite{DiFiGiSe91}, Proposition 1.12 and its proof we deduce that there exists a polynomial $Q\in\QQ[Y]$ of degree at most $d^{\cO(m^2)}$ such that any point $y\in \CC^{d\times n}$ with $y=(y_1,...,y_d)$, $y_1,...,y_d\in\CC^{n}$ and $Q(y)\neq0$ satisfies the condition of Corollary \ref{corEx1}.

Therefore, there exists such a point $y\in\ZZ^{d\times n}$ of logarithmic height at most $\cO(m^2\log d)$. This implies that the logarithmic heights of $G_1:=\Gamma(y,X),...,G_n:=\Gamma(y,X)$ and $H(U,y_1),...,H(U,y_d)$ are bounded by $(m^2+h)(nd)^{\cO(n)}$. We may now apply Theorem \ref{Thm21} to this situation to conclude that there exists a database $\cal D$ represented by an algebraic computation tree of size $(mL)2^{hd^{\cO(m^2+n)}}$ such that for any point $u\in[0,1]_{\Delta}$ the query whether the polynomial equation system $G_1(X)=0,...,G_n(X)=0, H(u,X)=0$ has a complex solution can be evaluated in the database $\cal D$ performing $hd^{\cO(m^2+n)}$ comparisons and $mL$ arithmetic operations in $\RR$.

\subsection{Example 2}
For $0\leq j< 2^n$ we write $[j]\in \{0,1\}^n$ for the representation of $j$ by $n$ bits.

Let $G_1:=X_1^2-X_1,...,G_n=X_n^2-X_n$ and $$H:=\sum_{1\leq k \leq m}\prod_{1\leq l\leq m}(1+(U_k^{2^l}-1)X_l).$$

Then, we have for $0\leq j< 2^n$ the identity $H(U,[j])=\sum_{1\leq k \leq m}U_k^j$. The degree and the logarithmic height of $H(U,[j])$ are bounded by $2^n$  and one, respectively. Furthermore $H$ may be evaluated by a straight-line program in $\ZZ[U,X]$ of size $L=\cO(mn)$. For $1\leq r\leq m$ and $1\leq i_1<...<i_r\leq 2^n$ let
$$\Delta_{i_1,...,i_r}:=\det \left( \begin{array}{ccc}
U_1^{i_1-1} & ... & U_r^{i_1-1} \\
\vdots & \ddots & \vdots \\
U_1^{i_r-1} & ... & U_r^{i_r-1} \end{array} \right)$$
and observe 
$\Delta_{i_1,...,i_r}\neq 0$. Thus $$\Delta:=\prod_{\begin{array}{c}1\leq r\leq m \\ 1\leq i_1<...<i_r\leq 2^n\end{array}}\Delta_{i_1,...,i_r}$$ is a non-zero polynomial of $\ZZ[U]$ of degree and logarithmic height at most $2^{\cO(mn)}$ .

For $1\leq r\leq m$ and $1\leq i_1<...<i_r\leq 2^n$, the Jacobian of $H(U,[i_1]),...,\allowbreak H(U,[i_r])$ is $$ \left( \begin{array}{ccc}
i_1U_1^{i_1-1} & ... & i_1U_m^{i_1-1} \\
\vdots & \ddots & \vdots \\
i_rU_1^{i_r-1} & ... & i_rU_m^{i_r-1} \end{array} \right).$$

Hence, the hypersurfaces of $\CC^m$ defined by the polynomials $H(U,[i_1]),...,\allowbreak H(U,[i_r])$ intersect transversally at any common point $u\in\CC^m_{\Delta}$.

We may again apply Theorem \ref{Thm21} to this situation to conclude that there exists a database $\cal D$ represented by a computation tree of size $$(m L) 2^{2^{\cO(m^3n^2)}}=2^{2^{\cO(m^3n^2)}}$$ such that for any point $u\in[0,1]^m_{\Delta}$ the query whether the polynomial equation system $X_1^2-X_1=0,...,X_n^2-X_n=0, H(u,X)=0$ has a complex solution can be evaluated in the database $\cal D$ performing ${2^{\cO(m^3n^2)}}$ comparisons and $mL=\cO(m^2n)$ arithmetic operations in $\RR$.

\paragraph{Observation:} Suppose that the ceiling function is available at unit costs. Then, the queries of Theorem \ref{Thm21}, Example 1 and Example 2 can be evaluated using $mL+n$ and $\cO(m^2n)$ arithmetic operations in $\QQ$, respectively.

The aim of Theorem \ref{Thm21} and Examples 1 and 2 is not to promote new upper complexity bounds for the queries under consideration. Our complexity model is purely algebraic and our algorithms cannot be translated efficiently to the bit model, at least not with the actual knowledge about the relationship between these two models.

However, Examples 1 and 2 show that, admitting branchings in our complexity model, the number of algebraic operations necessary to answer the queries may drop dramatically with respect to traditional methods based on the evaluation of elimination polynomials or their coefficient representation. In Theorem \ref{Thm21}, Example 1 and Example 2 these bounds are of order $Ld^n$, $Ld$ and $\cO(mn2^n)$, respectively (see \cite{HKR11} for a detailed discussion of this issue).

\section*{Acknowledgments}
Research partially supported by the following Argentinian, Belgian and
Spanish grants: CONICET PIP 2461/01, UBACYT 20020100100945, UBACYT 20020100300067, PICT-2010-0525, FWO G.0344.05, MTM2010-16051.

\bibliographystyle{amsplain} 
\bibliography{Tesis-Doc}

\end{document}